\documentclass[runningheads]{llncs}
\usepackage{color,soul}
\usepackage{changepage}
\usepackage{amssymb}
\usepackage{amsmath}
\usepackage{xcolor,colortbl}
\usepackage{graphicx}
\usepackage{algorithm}
\usepackage[noend]{algorithmic}
\usepackage{comment}
\newcommand{\mpara}[1]{\smallskip\noindent{\bf #1}}
\newcommand{\triple}[3]{(#1, #2, #3)}
\usepackage{float}
\begin{document}
\title{iSummary: Workload-based, Personalized summaries for  Knowledge Graphs}
\author{Giannis Vassiliou\inst{1} \and
Fanouris Alevizakis\inst{2,3} \and
Nikolaos Papadakis\inst{1} \and
Haridimos Kondylakis\inst{3}
}
\institute{
Department of Electrical and Computer Engineering, HMU \and 
Computer Science Department, UOC \and
Institute of Computer Science, FORTH}
\maketitle              % typeset the header of the contribution
\begin{abstract}
The explosion in the size and the complexity of the available Knowledge Graphs on the web has led to the need for efficient and effective methods for their understanding and exploration. Semantic summaries have recently emerged as methods to quickly explore and understand the contents of various sources. However, in most cases, they are static, not incorporating user needs and preferences, and cannot scale. In this paper, we present iSummary, a novel, scalable approach for constructing personalized summaries. As the size and the complexity of the Knowledge Graphs for constructing personalized summaries prohibit efficient summary construction, in our approach we exploit query logs. The main idea behind our approach is to exploit knowledge captured in existing user queries for identifying the most interesting resources and linking them, constructing as such high-quality, personalized summaries. We present an algorithm with theoretical guarantees on the summary’s quality, linear in the number of queries available in the query log. We evaluate our approach using three real-world datasets and several baselines, showing that our approach dominates other methods in terms of both quality and efficiency.

\end{abstract}

\section{Introduction}
Daily, a tremendous amount of new information becomes available online. RDF Knowledge graphs (KGs) rapidly grow to include millions or even billions of triples that are offered through the web. For example, the Linked Open Data
Cloud, currently includes more than 62 billion triples, organized in large and complex RDF data graphs~\cite{LODCloud}. 

%Ontologies can be used to connect to
%RDF data graphs by describing the possible classes their properties ,as well
%as the relationship between these classes and their properties.On one hand
%ontologies do provide an extra entry point into the data so to have an conceptual structure on other hand they can be very complex.

The complexity and the size of those data sources limit their exploitation potential and necessitate effective and efficient ways to explore and understand their content~\cite{DBLP:conf/sigmod/TroullinouKLM21}. In this direction, semantic summarization has been proposed as a way to extract useful, minimized information out of large semantic graphs that many applications can exploit instead of the original data graphs for performing certain tasks more efficiently such as visualization \cite{DBLP:conf/esws/PappasTRKP17}, exploration \cite{DBLP:conf/semweb/TroullinouKSP18}, \cite{DBLP:conf/semweb/TroullinouKSP18a}, query answering, etc. \cite{DBLP:journals/vldb/CebiricGKKMTZ19}. Structural semantic summaries focus mostly on the structure of the graph for extracting the required information, whereas non-quotient structural semantic summaries try to select the most important parts of the graph for generating the result summaries.

\textbf{The problem.} Most of the existing works in the area of structural, non-quotient semantic summarization, produce generic static summaries \cite{DBLP:journals/vldb/CebiricGKKMTZ19} that cannot be applied to big KGs. Further, as different persons have different data exploration needs the generated summaries should be tailored specifically to the individual's interests. Although this has already been recognized by the research community, the approaches offering personalized summaries so far, rely on node weights selected by the users, then followed by algorithms making various vague assumptions about the relevant subsets out of the semantic graph that should complement the initial user choice \cite{alzogbi2013similar} \cite{wu2008identifying}. More recent approaches like \cite{DBLP:conf/icdm/SafaviBFMMK19} exploit the individual user queries for mining user preferences but still rely on the KG to compute the summary which makes it computationally hard. Further, capturing a complete individual user query set is usually not feasible.

\textbf{The solution.} Instead of relying on node weights or on individual provided set of user queries, we exploit generic logs already available through the SPARQL endpoints of the various KGs available online. Then in order to generate a personalized summary we only require one or a few nodes the user is most interested in. As previous users have already identified through their queries, the most common connections to the specific user-selected nodes, we exploit this information in order to formulate the generated summaries. More specifically:

\begin{itemize}
  \item We introduce, motivate and formulate the problem of $\lambda/\kappa$-Personalized Summary and we show that although a solution to the problem is rather useful, resolving the problem is both impractical (requires multiple weights assignments) and computationally expensive (NP-complete).

  \item We analytically show how we can resolve the problem relying on existing query logs and we provide a solution to both the multiple weight assignment required and also to the computational problem.

  \item We present an algorithm that provides theoretical guarantees on the summary's quality which is linear in the number of queries available in the query log.
    
  \item We experimentally evaluate our approach using three real-world datasets and the corresponding workloads, i.e. Wikidata, Bio2RDF, and DBPedia, showing the benefits of our approach maximizing coverage for user queries, dominating all baselines and competitors on both quality and efficiency.

%  \item For example firstly we run as input a start node and 1 most frequent k-top node and we scan all queries datasets and for each query which contain this start and most frequent node we run BFS and calculate its shortest path.Until the end of querie dataset we find all shortest path with this start and top-k end node,Finally if we have find 5 shortest paths with start nodes and this top-k node we store the shortest from all.
 % \item We do the same for each other top k node until last top k node.
%\item Finally if user has give as input top-5 nodes frequent we have 5 shortest paths.
%\item In next station we do our expirements calculating compression ratio which is sum of percentage of most frequent nodes and all nodes in each query each time.
\end{itemize}

To the best of our knowledge, this is the first approach to constructing personalized, structural, non-quotient semantic summaries exploiting generic query workloads. 
%And althought query workloads are hard to mine, the KG owners definitely have them available and can use them for helping their user-base to explore through summarization the contents of these KGs.
The rest of this paper is organized as follows:
Section 2 provides preliminaries and problem definition. Then Section 3 presents our solution, iSummary, detailing the various steps for generating a personalized summary.
Section 4 presents the experimental evaluation of our work, whereas Section 4 presents related work.
Finally, Section 5 concludes this paper and presents directions for future work.

\section{Preliminaries \& Problem Definition}

\mpara{Preliminaries.} In this paper, we focus on RDF Knowledge Graphs, as RDF is among the most widely-used standards for publishing and representing data on the Web, promoted by the W3C for semantic web applications. An \textit{RDF KG} $\mathcal{G}$ is a set of \textit{triples} of the form $\triple{s}{p}{o}$.
A triple states that a {\em subject} $s$ has the {\em property} $p$, and the value of that property is the {\em object} $o$. 
We consider only well-formed triples, according to the RDF specification~\cite{W3C-RDF}. These belong to $(\mathcal{U} \cup \mathcal{B}) \times \mathcal{U} \times (\mathcal{U} \cup \mathcal{B} \cup \mathcal{L})$, where $\mathcal{U}$ is a set of Uniform Resource Identifiers (URIs), $\mathcal{L}$ a set of typed or untyped literals (constants), and $\mathcal{B}$ a set of blank nodes (unknown URIs or literals); $\mathcal{U},\mathcal{B},\mathcal{L}$ are pairwise disjoint. Additionally, we assume an infinite set $\mathcal{X}$ of variables that is disjoint from the previous sets.
Blank nodes are essential features of RDF allowing to support {\em unknown URI/literal tokens}. The RDF standard includes the \texttt{rdf:type} property, which allows specifying the type(s) of a resource. Each resource can have zero, one or several types. For querying, we use SPARQL~\cite{w3csparql}, the W3C standard for querying RDF datasets. 
The basic building units of the SPARQL queries are triple pattern and Basic Graph Pattern (BGP). A triple pattern is a triple from $(\mathcal{U} \cup \mathcal{B} \cup \mathcal{X}) \times (\mathcal{U} \cup \mathcal{X})) \times (\mathcal{U} \cup \mathcal{B} \cup \mathcal{L} \cup \mathcal{X})$. A set of triple patterns constitutes a basic graph pattern (BGP).

\mpara{Informal problem statement.} 
Informally the problem we address may be described as follows: Given a knowledge graph $G$, a limited set of $\lambda$ resources that the user wants his/her summary to be focused on,  and a number $\kappa$ denoting the size of the summary (in terms of nodes to be included), efficiently construct a personal summary $G\prime \in G$ that best captures the user's preferred information in G. 

Resolving this problem is really important, as usually users visit a KG with a specific information request in mind, and are used in providing a starting point to begin the KG exploration that will lead to the information they are looking for. Usually, they are not interested in generic summaries of the overall graph, but they would like to identify information pertinent to a specific part of the graph \cite{DBLP:conf/ssdbm/VassiliouTPK21}.

\begin{figure}[htp]
    \centering
    \includegraphics[width=\textwidth]{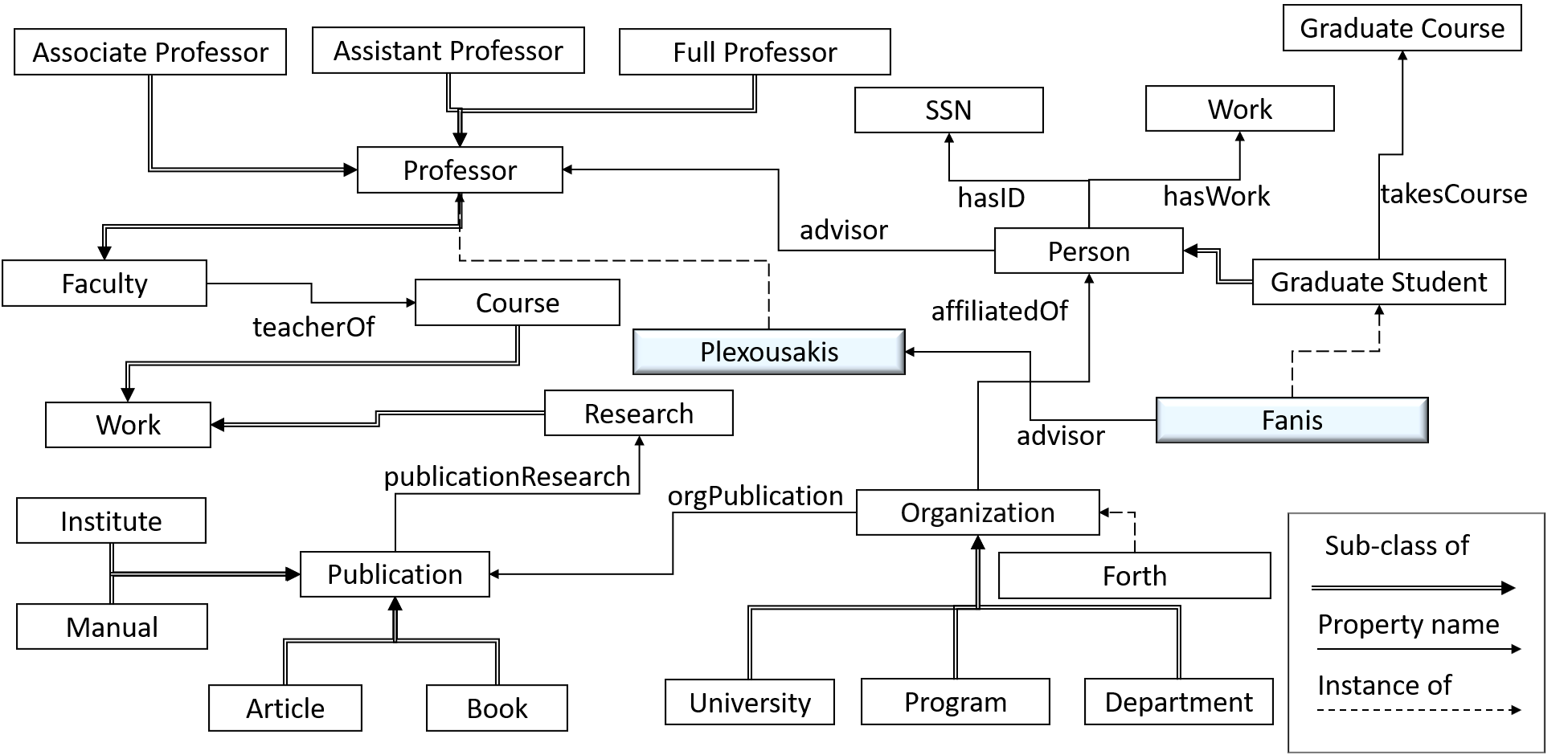}
    \caption{Example RDF KG.}
    \label{fig:examplerdf}
\end{figure}

\begin{example}
Consider as an example, the KG shown in Figure~\ref{fig:examplerdf}, which includes information on the university domain. The figure visualizes persons and organizations and also presents some indicative instances. Note that prefixes are omitted from the figure for sake of clarity. Now assume that the user selects two nodes ($\lambda$=2), i.e. $Plexousakis$ and $Fanis$ (the blue ones), and would like to get a personalized summary of size five ($\kappa$=5). As such, three more nodes should be selected from the graph and linked with the two nodes provided by the user. 

A way to select the three additional nodes and the edges for the result summary, used by previous approaches (e.g., \cite{zhang2007ontology}) is to have weights available on all (or some of) the nodes, and select the nodes maximizing the weight of the selected sub-graph. However, those weights should be specific to user requests. For example, $Publication$ might not be of interest when requesting a summary for $Plexousakis$ and $Fanis$ whereas it might be of interest when requesting a summary for $Research$.
\end{example}

%\HK{more argumentation on why this is useful}

\mpara{Formal problem statement.} 
The previous example makes it obvious that requesting the user to provide weights each time for all (or a least some) of the nodes is impractical. Next, we formally present the problem of $\lambda/\kappa$-Personalized Summary and we show that, although useful, besides impractical it is also computationally expensive.

\begin{definition}[$\lambda/\kappa$-Personalized Summary] 
Given 
(1) a knowledge graph $G=(V,E)$,
(2) a non-negative weight assignment to all nodes, capturing user preferences in $G$, 
(3) $\lambda$ seed nodes,
(4) and a number $\kappa$ ($\lambda \le \kappa$), 
find the smallest maximum-weight tree $G\prime=(V\prime, E\prime) \in G$ including the $\kappa$ most preferred nodes.
\end{definition}

Not that we don't actually require a weight to be assigned to all nodes, as the weight of all nodes can be by default zero, and the user only adds weights to a subset of them. A solution to the $\lambda/\kappa$-Personalized Summary problem is not unique, as there might be many maximum-weight trees with the smallest size that are equally useful for the user. Next, we prove that the aforementioned problem is NP-complete.

\begin{theorem}
The $\lambda/\kappa$-Personalized Summary problem is NP-complete.
\end{theorem}
\begin{proof}
The Steiner tree problem \cite{DBLP:journals/networks/Hakimi71}, focuses on connecting selected nodes of a weighted graph at minimum cost. In our case, we normalize weight assignments from 0 to 1 and subtract them from 1. Further, we set the weight of the $\lambda$ seed nodes to be equal to zero. Now instead of finding a maximum-weight tree, we search for a minimum-weight tree, connecting the seed nodes with the $\kappa-\lambda$ minimum weight nodes. As such our problem is equivalent to the Steiner tree problem which has been shown to be NP-complete.
\end{proof}

A nice property of the $\lambda/\kappa$-Personalized Summaries is that their quality is \textit{monotonically increasing} as the $\kappa$ increases. This means that as the summary size increases more relevant information is added to the summary for the same seed nodes selected by the user.

\begin{lemma}
Let $S_{\kappa}$ be a $\lambda/\kappa$-Personalized Summary and $S_{\kappa+1}$ be a $\lambda/(\kappa+1)$-Personalized Summary for $G$. Then $W(S_{\kappa+1}) \ge W(S_{\kappa})$, where $W(S)$ the sum of all node weights in $S$.
\end{lemma}
\begin{proof}
As $S_k$ is a maximum-weight tree for $\lambda$ including $\kappa$ nodes, adding one more node in the summary and looking for the maximum-weight tree including that node as well guarantees that the total weight of $S_{k+1}$ will be equal or greater than the total weight of $S_k$.
\end{proof}

Over the years many approximations have been proposed for resolving the Steiner Tree problem \cite{DBLP:journals/networks/HwangR92,DBLP:journals/dam/Voss92} that could be exploited for resolving the $\lambda/\kappa$-Personalized Summary problem as well. 
CHeapeast INSertion (CHINS) one of the fastest approximation algorithms has a worst time complexity of $O(\kappa \times 2 \lvert V \rvert  \times log \lvert V \rvert)$. CHINS starts with a partial solution consisting of a single selected node and it incrementally adds the nearest one of the selected not yet in the solution. However, still, computing a Steiner Tree approximate solution over commodity hardware for a large KG such as WikiData is not feasible. For example, assuming 1$\mu$s for each operation, running CHINS for WikiData that includes 1.4 billion statements would require more than a year to calculate a 5/10-Personalized Summary. 

For the rest of the paper, without loss of generality, we will focus on 1/$\kappa$-Personalized Summaries (in short $\kappa$-Personalized Summaries), where the user provides only a single seed node as input, for not perplexing definitions and algorithms and due to space limitations. Extending the presented solution and algorithms for multiple seed nodes is straightforward.

\section{iSummary}
As we have shown in the previous section, computing the $\kappa$-Personalized Summary is both \textit{impractical}, as different weights should be assigned to the graph nodes for each distinct user query, and computationally \textit{expensive}, as it requires computing a Steiner Tree solution. In this section, we are going to provide an elegant approximate solution based on query workloads.

\mpara{Resolving the problem of multiple weight assignments.} 
Assume now that for the KG $G$ we have available a query log $Q=\{q_1, \cdots, q_n\}$ available. This assumption is reasonable, as all big KGs offer a SPARQL endpoint that logs user queries for various purposes. Multiple studies already confirm this (e.g., \cite{DBLP:journals/vldb/BonifatiMT20}), and we were also able to easily get access to such logs for DBpedia, WikiData, and Bio2RDF (more about this in Section~\ref{sec:evaluation}).

Having such a query log available, our first idea is that we can \textit{use it to mine user preferences} for the specific seed node that the user is interested in. The idea here is that if a user is interested in a $\kappa$-Personalized Summary for $s$ then we can use $Q$ to identify \textit{relevant queries} to $s$, i.e., queries that include $s$. In those queries, other nodes relevant to the user input will be available. In fact, as those queries have been issued by thousands of users, we assume that \textit{the most useful related nodes will be the ones that appear more frequently} there.
 
\begin{example}
Assume that for our example KG, shown in Figure \ref{fig:examplerdf}, we have available a query log consisting of the following SPARQL queries:
\begin{verbatim}
Q1. SELECT ?x ?y WHERE 
    {x? a Person. y? a Professor. ?x advisor ?y.}
Q2. SELECT ?x ?y WHERE 
    {x? a Person. y? a Organization. ?y affiliatedOf ?x.}
Q3. SELECT ?x ?y WHERE 
    {x? a Person. y? a Organization. ?y affilatedOf ?x. 
    ?y orgName "FORTH".}
Q4. SELECT ?y WHERE 
    {y? a Organization.}
Q5. SELECT ?y WHERE 
    {y? a Publication. ?x authored ?y. ?x a Institute.}
\end{verbatim}
Now assume that a user is interested in a 2-Personalized Summary for the node $Person$. Based on the query log we can identify that relevant queries to user input are $Q1$, $Q2$, and $Q3$. Examining those queries we can identify that the useful nodes are the $Professor$ and $Organization$. In fact, as $Organization$ is used in two queries it should be most useful according to the available query log. As we are looking for a 2-Personalized Summary it will be included in the result. On the other hand, if the user is interested in a 2-Personalized Summary for the node $Publication$ the relevant query is $Q5$ which suggests that the $Institute$ node should be included in the result.
\end{example}

Based on this assumption we can have multiple weight assignments, one per user input, as they occur from thousands of user queries that involve the provided user input and that are based on past users' preferences, as expressed in their queries. Note here that we don't need weights for the whole graph, as by default we can set the weight of the nodes that do not appear in the filtered user queries to zero.

\mpara{Resolving the computational problem.} Now that we have a way to assign personalized weights to the nodes, we will provide a computationally efficient procedure in order to link the selected nodes over a big graph. We will stick to the ideas proposed by the CHINS approximation algorithm. We will start with a solution including a single node, the $s$ selected by the user, adding one node each time of the ones with the maximum weight till all remaining $k-1$ nodes are included in the summary. However, for doing so we will not use the original data graph but again \textit{relevant user queries}. The main idea here is the following: link $s$ with the $k-1$ maximum weight nodes using \textit{the most frequent shortest paths} from the user queries. 

\begin{example}
We now continue our example for constructing a 2-Personalized Summary for the node $Person$. As we have already explained the node  $Organization$ has the higher frequency in queries involving $Person$ and as such it will be selected to be included in the summary. Now instead of searching the graph shown in Figure~\ref{fig:examplerdf} for linking $Person$ with $Organization$ we will additionally filter queries including $Person$ keeping only the ones including $Organization$ as well. Those are Q2 and Q3. For each one of those queries, we calculate the shortest path for linking $Person$ and $Organization$ and we eventually select the most frequent shortest path to include in the summary. As such the 2-Personalized Summary for the node $Person$ includes a single triple $t_1: \triple{Organization}{addiliatedOf}{Person}$. 

In the case we are interested in a 3-Personalized Summary for the node $Person$, the summary would have to include the $Professor$ node as well. To link $Person$ with $Professor$ we would filter the queries to keep only those where both $Professor$ and $Person$ appear, e.g. $Q1$. Then for linking those nodes, we would keep the most frequent shortest path, i.e., $t_2: \triple{Person}{advisor}{Professor}$. Now the 3-Personalized Summary for $Person$ would include both $t_1$ and $t_2$.

\end{example}

%\mpara{Exploiting queries for evaluating the summary quality.} Now instead of using user queries for constructing the summaries, assuming we also get 

\subsection{The algorithm}
Now we are ready to present the corresponding algorithm for constructing a $\kappa$-Personalized Summary for an input node $s$.
The algorithm is presented in Algorithm~\ref{algo} and receives as input $\kappa$, $s$ and a query log $Q$. 
It starts by including the first node in the summary (line 2), the one selected by the user. Then it filters the queries to keep only $Q_s$, i.e., the ones including $s$ (line 3).
Next, it calculates the frequency of all nodes in $Q_s$ and selects the $k-1$ ones with the higher frequency to be included in the result summary (line 4), i.e., the $top_{k-1}$ ones.

The next step is to visit one by one these nodes each time identifying an optimal way to link each one of those nodes not in the summary with the ones already added (lines 5-13).
More specifically for each node in $top_{k-1}$ not already in the summary we explore all nodes in the summary by filtering again the queries in $Q_S$ retrieving $Q_{sxy}$ that contains $x$ and $y$ (line 8). 
Then for each query in $Q_{sxy}$, we find the shortest path linking $x$ with $y$ (line 10) and we keep the most frequent one (line 11) to formulate the result summary. Eventually, we select to link the next node in the $top_{k-1}$ with the most frequent shortest path linking that node with all nodes currently in the summary. However, as we identify paths in the queries, those might include variables that we should replace with actual resources. This is accomplished by replacing them with resources mined from other queries which might have both the specific resource and its neighbors instantiated (line 13). Finally, we return to the user the constructed set of triples $S$ as a summary (line 14).

\begin{algorithm}
\caption{iSummary}
\label{algo}
\begin{flushleft}
\textbf{Input:} An user-selected node $s$, a query workload $Q$, the  number of the most useful nodes to be included in the summary $\kappa$. \\
\textbf{Output:} $S$ a $\kappa$-Personalized summary for $s$
\end{flushleft}
\begin{algorithmic}[1]
\STATE $ S \gets \emptyset$
\STATE $visited \gets \{s\}$
\STATE $ Q_s \gets filter(Q, \{s\})$      %\Leftcomment{\scriptsize{get the queries including $s$}}
\STATE $ top_{k-1}\gets selectTopNodes(Q_s, k-1)$
\FORALL{$x\in top_{k-1}, x\notin visited$}
    \STATE $selectedPath \gets \emptyset$
    \FORALL{$pairs(x,y), y \in visited$}
        \STATE $Q_{sxy} \gets filter(Q_s,\{x,y\}) $ %\Comment{\scriptsize{get the queries including both $s$ and $x$}}
        \FORALL{$q \in Q_{sxy}$}
            \STATE $shortestPaths[q] \gets getShortestPathFromQuery(q, \{x, y\})$
        \ENDFOR   
        \STATE $selectedPath \gets findMostFrequent(shortestPaths, selectedPath)$
    \ENDFOR
        
    \STATE $visited \gets visited \cup \{x\}$        
    \STATE $ S \gets S \cup resolveVariables(selectedPath, Q)$
\ENDFOR
\STATE \textbf{return}  $S$
\end{algorithmic}
\end{algorithm}

%The correctness of the algorithm is proved by the construction as in essence it is the CHINS algorithm

The result produced by the aforementioned algorithm is deterministic based on its implementation, as  in the case of ties, these are broken by keeping the first choice. However as already explained a personalized workload-based summary might not be unique as many nodes can have the same frequency in the available queries, or there might be available many different shortest paths to connect them. Next, we prove that iSummary is able to find an approximate solution with specific guarantees:

\begin{theorem}
The iSummary algorithm finds an approximate solution to the $\kappa$-Personalized Summary problem with a worst-case bound of 2, i.e., $W/W_{opt} \leq 2 \times (1-l/k)$, where $W$ and $W_{opt}$ denote the total weight of a feasible solution and an optimal solution respectively, and l a constant.
\end{theorem}
\begin{proof} (sketch)
In essence, iSummary replicates the CHINS approximation algorithm which has been proved to have the aforementioned worst-case bound \cite{hawking1988} using the queries in order to reconstruct the part of the interest of the original graph. For the remaining nodes that do not appear in the filtered queries, we set their weight to zero. The proof follows.
\end{proof}

To identify the complexity of the algorithm we should first identify the complexity of its components. Assuming $ \lvert Q \rvert$ the number of queries in the available workload we first need to scan them once for filtering and retrieving the $top_{k-1}$ nodes, i.e.  $O(\lvert Q \rvert)$. 
Then for each node in the $top_{k-1}$ we need to gradually include them in the visited set by checking their connection to all existing nodes in the summary. This will result in $k^2$ iterations in each of which the $Q_s$ queries should be filtered, i.e. $O( k^2 \times \lvert Q_{s} \rvert)$. Then for each query appearing in the filtering results we should run once the Dijkstra algorithm for getting the shortest path. At the worst case for each node we need to calculate the shortest paths for all queries, i.e. $O( k^2 \times \lvert Q_s \rvert \times \lvert V_{Q_s}^2 \rvert )$, where $V_{Q_s}^2$ the maximum number of nodes that appear in the queries in the workload. Overall the complexity of the algorithm is 
 \[ O(\lvert Q \rvert) + O( k^2 \times \lvert Q_s \rvert \times \lvert V_{Q_s}^2 \rvert ) \leq O( k^2 \times \lvert Q \rvert \times \lvert V_{Q}^2 \rvert ) \]    

However, usually, the number of nodes requested by the user to be included in the summary is small. Also, the number of nodes in the queries is limited (usually $\leq 10$), and as such we can safely replace $V_{Q}^2$ with a constant, eventually showing that the algorithm scales linearly to the number of queries in the workload.

\mpara{Limitations.} The aforementioned algorithm provides an elegant solution to the $\kappa$-Personalized Summary problem and can be trivially extended for the $\lambda/\kappa$-Personalized Summary problem as well, just by searching for queries including the $\lambda$ nodes and then exploiting those queries to link them in the summary. However, it assumes that \textit{adequate} queries are available in the query log. In other words, it assumes that a) there are queries available including user input, and b) that there are at least $\kappa$ other nodes available in those queries. These assumptions hold for popular online KGs which can easily log user queries but might not hold for other less popular KGs. As such our approach should be considered complementary to approaches working directly on the graphs of the KGs. However, as we showed the problem is NP-complete, and neither existing approximate solutions nor competitors (as we will show) will terminate within a reasonable time.

\section{Experimental Evaluation} 
\label{sec:evaluation}

In this chapter, we present the experiments performed for evaluating our approach using three real world datasets along with the corresponding query workloads.
The source code and guidelines on how to download the datasets and the workloads are available online\footnote{\url{https://anonymous.4open.science/r/iSummary-47F2/}}. 

\subsection{Setup}
\mpara{Implementation.} 
The iSummary was developed using Java. In addition, the evaluation was performed using windows 10 with an Intel\textsuperscript{\textregistered} Core\textsuperscript{TM} i3 10100 CPU @ 3.60GHz (4 cores) and 16 GB RAM.

%The high-level workflow of our approach is shown in Figure~\ref{fig:galaxy1}. As shown iSummary gets as input a) a query workload for that KG; c) an initial selection of resources, and d) the size of the requested summary in terms of the most important nodes to include. Then iSummary selects the $k$ most important nodes for the user and it links them based on the query workload. Eventually it presents a connected graph to the user. For finding the top-$k$ nodes and for linking them the query workload is explored by finding the shortest paths out of all queries involving the nodes of interest. 

%\begin{figure}[htp]
%    \centering
%    \includegraphics[width=\textwidth]{figures/logicalgorithm.png}
%    \caption{Workflow of our summary technique.\HK{@giannis: to improve figure}}
%    \label{fig:galaxy1}
%\end{figure}

\mpara{Datasets.}  
The first dataset we use is DBpedia v3.8 along with the corresponding query workload. DBpedia v3.8 consists of 422 classes, 1323 properties, and more than 2.3M instances. The available query workload is 16.3 MB including 58,610 queries. 

WikiData is a free and open knowledge base that can be read and edited by both humans and machines. Wikidata contains 100 million items, and 1.4 billion statements and covers many general topics for knowledge exploration and data science applications. The query workload for WikiData was retrieved from~\cite{DBLP:conf/semweb/MalyshevKGGB18} and includes 192,325 queries
 
Bio2RDF is a biological database that uses semantic web technologies to provide interlinked life science data and includes more than 11 billion triples \cite{DBLP:conf/semweb/DumontierCCAEBD14}.
The query workload for Bio2RDF was retrieved from the corresponding SPARQL endpoint and includes 3,616,330 queries.

\subsection{Metrics} 
We already have proven the theoretical bound of our algorithm in terms of quality when compared to an optimal solution. In addition, as it is not feasible to compute the optimal solution for our big graphs for evaluating the quality of the generated algorithms
%, we initially \textbf{evaluate node selection}, in order to identify how good each approach is on selecting the most important nodes for a given user input. 
%For a more objective metric, measuring both the nodes and edges of a result summary 
we use \textbf{coverage}. Coverage has been proved rather useful in evaluating structural, non-quotient semantic summaries in the past
\cite{DBLP:journals/algorithms/TrouliPTKPK23},
\cite{DBLP:conf/semweb/TrouliTKPK21}, 
\cite{DBLP:conf/semweb/TroullinouKSP18},
\cite{DBLP:conf/ssdbm/VassiliouTPK21},
\cite{DBLP:confVassiliouTPSPK21}. The idea behind coverage is that, ideally, we would like to maximize the fragments of the queries that are answered by the summary. More specifically, a summary that is able to provide answers to bigger and more query fragments from the query workload is preferable.
However, as we are generating personalized summaries, we would like the generated summaries to maximize the number and fragments that include the input provided by the user. As such, we define coverage as follows:

%
%having a summary, we can calculate for each query, the percentage of the nodes and the edges that are included in the summary.
%This means that for example if user get as answer from summary 10 shortest path we get each path and getting the nodes from each path we compare in each dataset each query nodes how many shortest path node have and after divide with total query nodes (check type 5.1 where.This procedure comparison becomes for each query in query dataset and for each summary path.Finally for all shortest path and all queries in query dataset we divide with number n of queries in workload dataset so have total Coverage in all query workload.\newline So lets assume a query workload  Q=\{q1...qn\} we define coverage as follows 

%\begin{equation}
%\label{OverallPrecision2}
% Coverage(Q, S) = \frac{\frac{\sum_{i=1}^{v} s_{nodes}(qi)}{nodes(qi)}}{n}
%\end{equation}

\begin{definition}[Coverage]  Assuming a $\kappa$-Personalized summary $S$ for $s$, a query workload $Q$, and two weights for nodes and edges, i.e. $w_n$ and $w_p$, we define coverage as follows:
	\begin{equation}
Coverage(Q, S, s)= \frac {1}{n} \sum_{s \in q_i } (w_n \frac{ snodes(S, q_i)}{nodes(q_i)} + w_p \frac{sedges(S, q_i)}{edges(q_i )})
	\end{equation}
	%where $\{q_1, \cdots, q_n\}$ the queries in $Q$ that include $s$.

 where $nodes(q_i)$ and $edges(q_i)$ denote the number of nodes and edges respectively in $q_i$, and $snodes(S, q_i)$ and $sedges(S, q_i)$ denote the number of nodes and edges respectively that appear in $S$.
\end{definition}
	
In our experiments, we set $w_n$=0.5 and $w_p$ = 0.5 as we perceive both nodes and edges as equally important in a summary.
    %in essence our summaries are node-based and we prioritise node selection \HK{@giannis: is this what we use?}.
	
%\begin{equation}
%\label{OverallPrecision}
% EachQueryCoverage = \frac{\sum_{i=1}^{v} N_i}{m}
%\end{equation}

%\begin{equation}
%\label{OverallPrecision}
 %TotalCoverage = \frac{\frac{\sum_{i=1}^{v} N_i}{m}}{k}
%\end{equation}

%In our case we did this calculation for 2 workloads queries datasets 1 for wikidata and 1 for dbpedia

\subsection{Baselines \& Competitors} 
To evaluate our system, we use for each query workload a percentage of the queries for constructing the personalized summary (train queries) and the remaining queries for evaluating node selection and coverage of the constructed summary (test queries).

We compare our approach with a \textit{random} baseline, where we randomly select nodes and edges from the train queries that involve user selection to be included in the summary and then evaluate node selection and coverage over the test queries.

In addition, we compare our approach with another summarization method for personalized summaries \textit{GLIMPSE} \cite{DBLP:conf/icdm/SafaviBFMMK19} which tries to maximize a user’s inferred “utility”
over a given KG, subject to a user- and device-specific constraint on the summary’s size.

Finally, we explore an approximate version of the \textit{personalized PageRank}~\footnote{https://github.com/asajadi/fast-pagerank} which works directly on the KG trying to identify the most important nodes and paths given a start node through random walks.

\subsection{Coverage for Various Query Log Sizes} 
In the first experiment, we try to understand what is the size of the query log required for getting high-quality results in terms of coverage for iSummary. As such, we keep a random 20\% of the queries for testing and we use the remaining for the training. We gradually increase the percentage of queries considered and report the average coverage each time. We randomly pick a node to be used as a seed node for construction summaries for k=5, 10, and 15. We repeat the experiment 10 times (10 fold-cross validation). The results are shown in Figure~\ref{fig:queriessize}.

\begin{figure}[t]
    \centering
    \includegraphics[width=\textwidth]{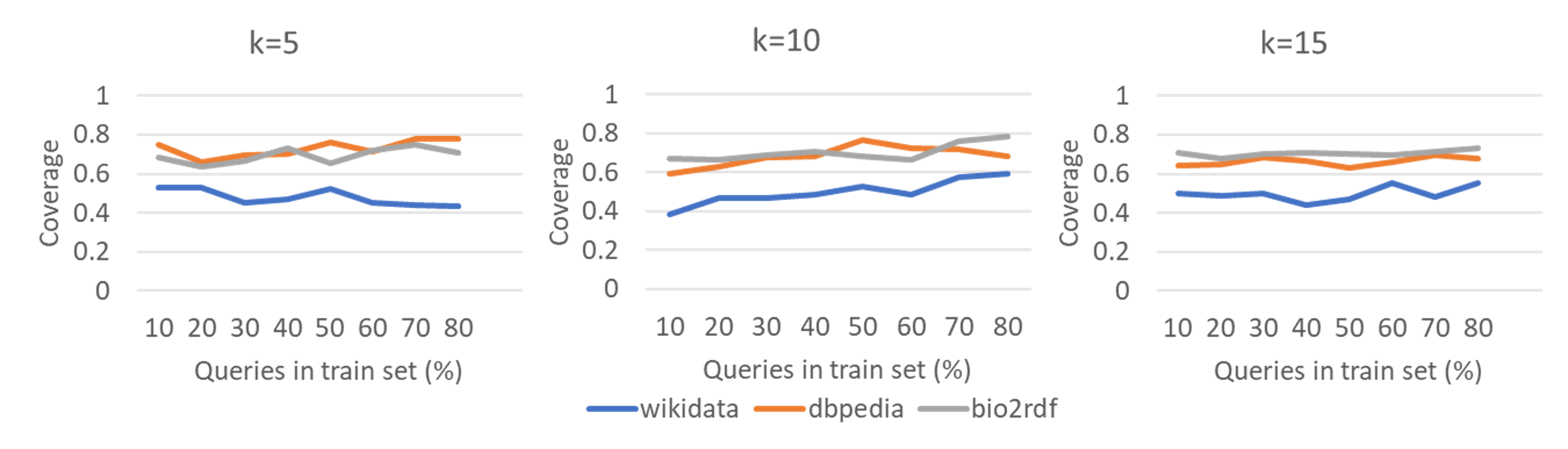}
    \caption{Coverage as the number of queries increases.}
    \label{fig:queriessize}
\end{figure} 

As shown, query coverage even with 10\% of the queries (i.e., ~6000 queries for DBpedia) is more than 0.4. Further it is not significantly increased as more queries are considered for constructing the summary. This  shows that our method is able to generate high coverage summaries even with relatively small size of queries. 

In addition, we can see that the worst coverage is for  wikidata. Trying to identify the reason for this we identified that in Wikipedia on average the queries include 3.9 triple patterns, in Bio2RDF 1.5 triple patterns, and in DBpedia 1 triple pattern.
%, showing that our algorithm is sensitive to the number of triple patterns included in the user queries as 
Based on this we can conclude that larger queries introduce more nodes on average for coverage evaluation (in the denominator of equation (1)), and as such the coverage drops.
%as this introduces more variables that should be resolved.

\subsection{Comparing Coverage}
Next, we compare iSummary with baselines and competitors. For iSummary and Random we randomly select 80\% of the queries for training and 20\% for testing. For the same test queries each time we evaluate also coverage for PPR and GLIMPSE. We randomly select 10 seed nodes for generating a personalized summary for k= 5, 10, 15. We repeat 10 times the aforementioned procedure (10 fold cross-validation). 

 \begin{figure}[H]
    \centering
    \includegraphics[width=0.8\textwidth]{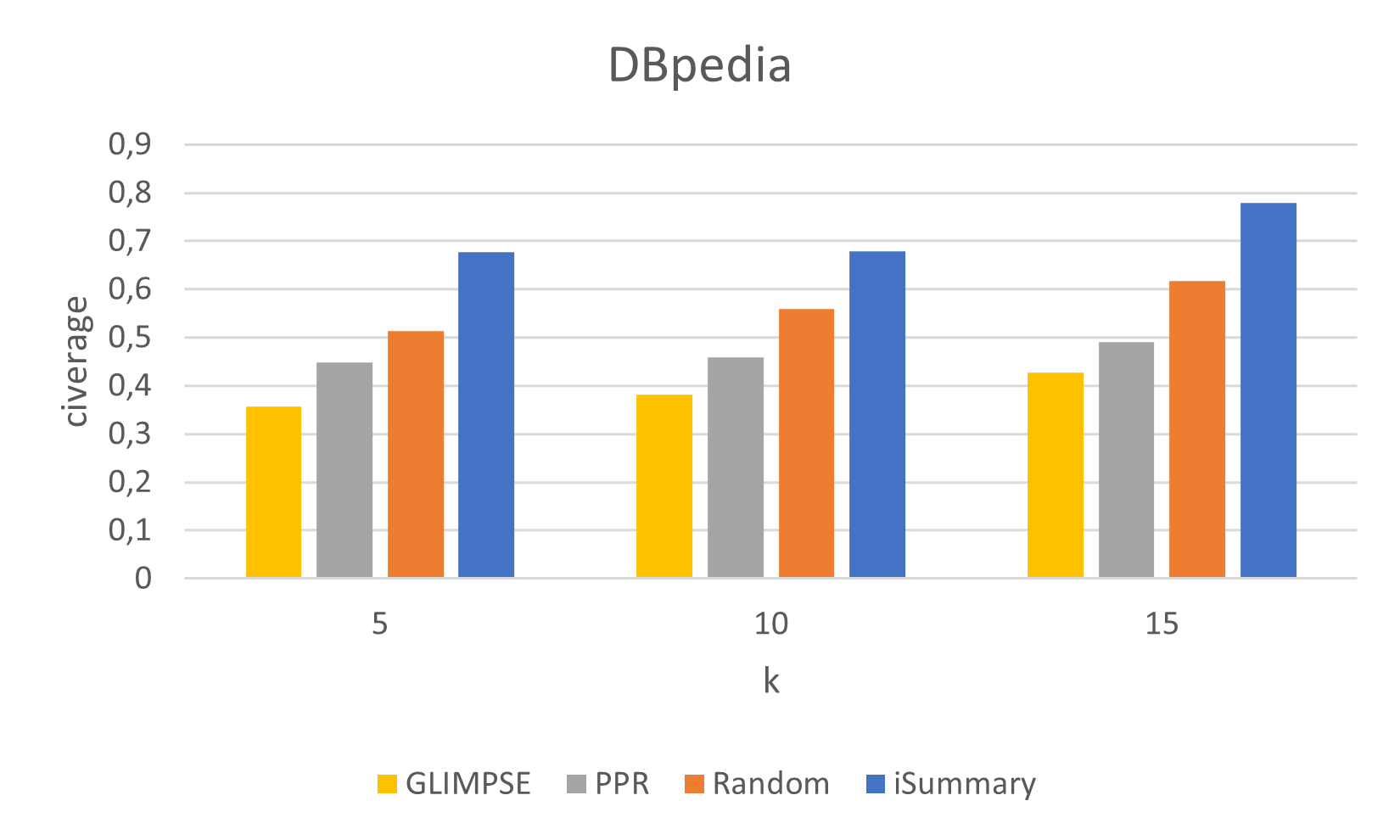}
    \caption{DBpedia coverage for various $k$ and baselines.}
    \label{fig:reultsdbpedia}
\end{figure} 

The results for DBpedia are shown in Fig.~\ref{fig:reultsdbpedia}. As shown approaches that work on the data graph have worst coverage than the ones working directly on the queries. GLIMPSE performs worst for all cases, as providing just a node as an input is not enough for GLIMPSE to provide a high-quality summary in terms of coverage. PPR has better results than GLIMPSE, but still, it is outperformed by both Random and iSummary. Note that Random is not purely random as it randomly selects nodes and edges to construct a summary from the queries involving the input node. As shown iSummary outperforms all baselines almost two times when compared with GLIMPSE, random by 17-24\% and PPR by 32-37\%.

\begin{figure}[H]
    \centering
    \includegraphics[width=0.8\textwidth]{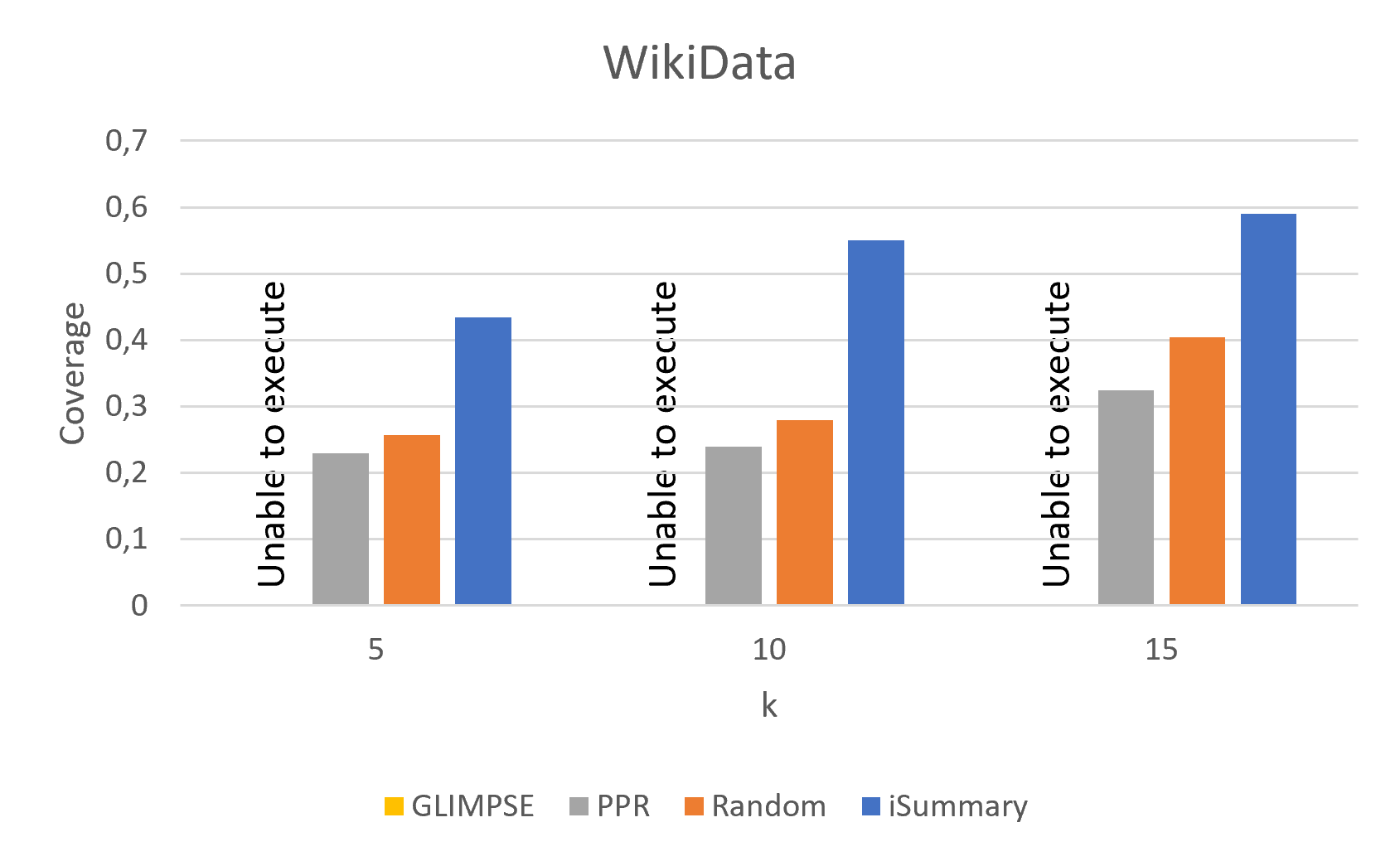}
    \caption{Wikidata coverage for various $k$ and baselines.}
    \label{fig:reultswikidata}
\end{figure}

The same trend appears for WikiData as shown in Fig.~\ref{fig:reultswikidata}. GLIMPSE is not able to produce output for such a big graph in our machine as it fully loads the memory and the application crashes after some time. In this case, iSummary dominates the remaining baselines, achieving in most of the cases a two times higher coverage.

\begin{figure}[H]
    \centering
    \includegraphics[width=0.8\textwidth]{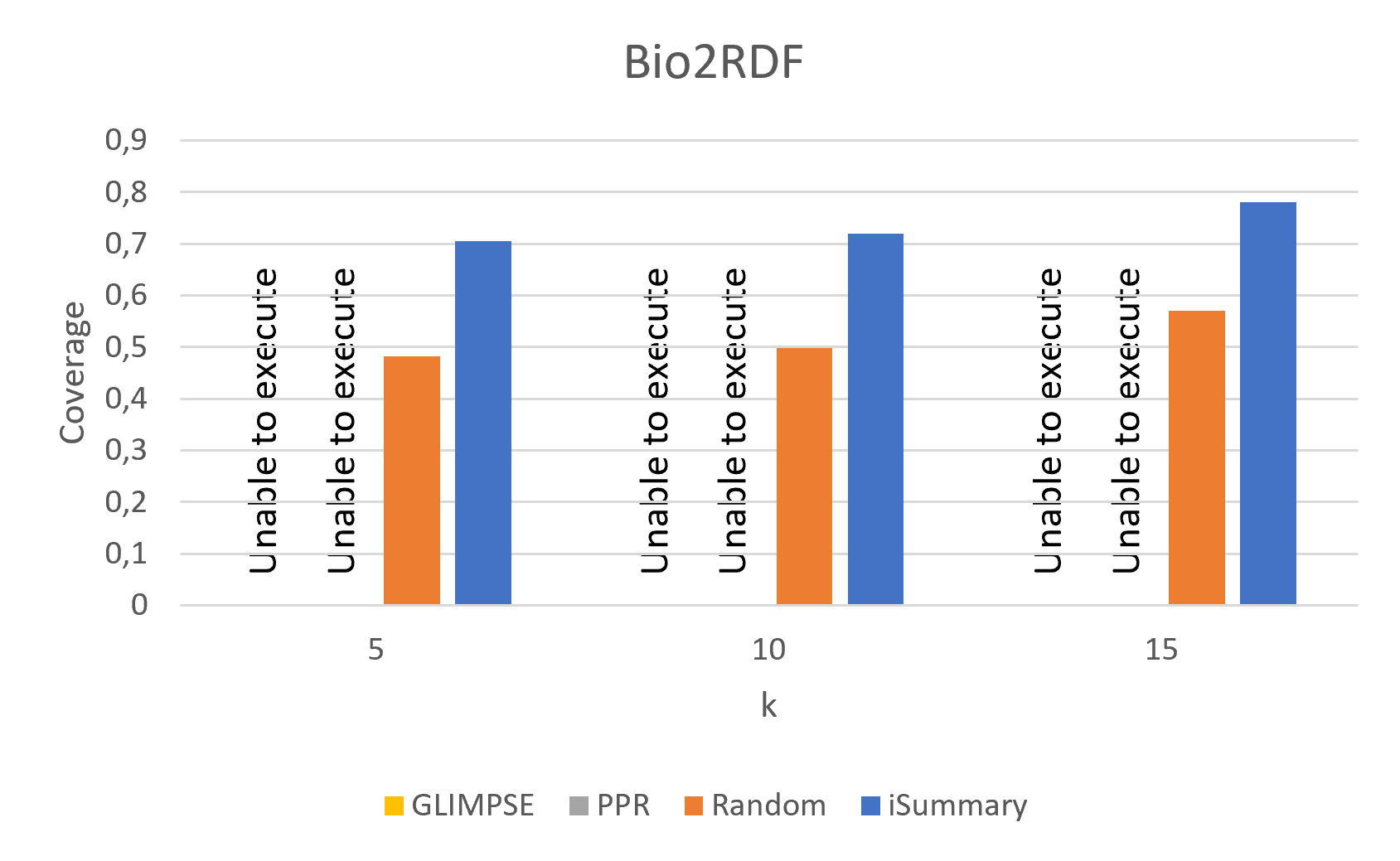}
    \caption{Bio2RDF coverage for various $k$ and baselines.}
    \label{fig:reultsbio2RDF}
\end{figure}

Finally, results for Bio2RDF are shown in Fig.~\ref{fig:reultsbio2RDF}. Now even PPR cannot process such a big graph and after 24 hours and we stopped its execution. Again iSummary is better than Random by 25-30\%.

Overall, as we can see in all cases our approach has consistently better results than all baselines, demonstrating the high quality of the generated summaries. We can also notice that as the size of the personalized summary increases ($\kappa$=5,10,15) the coverage increases as well, as more nodes are added to the summary. 
%However, just adding additional nodes is not enough, as only adding high-quality nodes increases the overall coverage. This is evident in our approach, as randomly adding more nodes marginally improves coverage only in some cases. 
Note also that as the size of WikiData queries is larger than the other datasets it is reasonable to be a bit more difficult to cover them and as such coverage is smaller.
Nevertheless, the algorithm shows stability among different datasets always dominating other approaches.
%reaching a coverage of around 0.5 for k=15, meaning that on it covers a significant number of relevant query fragments.

\begin{figure}[h]
    \centering
    \includegraphics[width=0.8\textwidth]{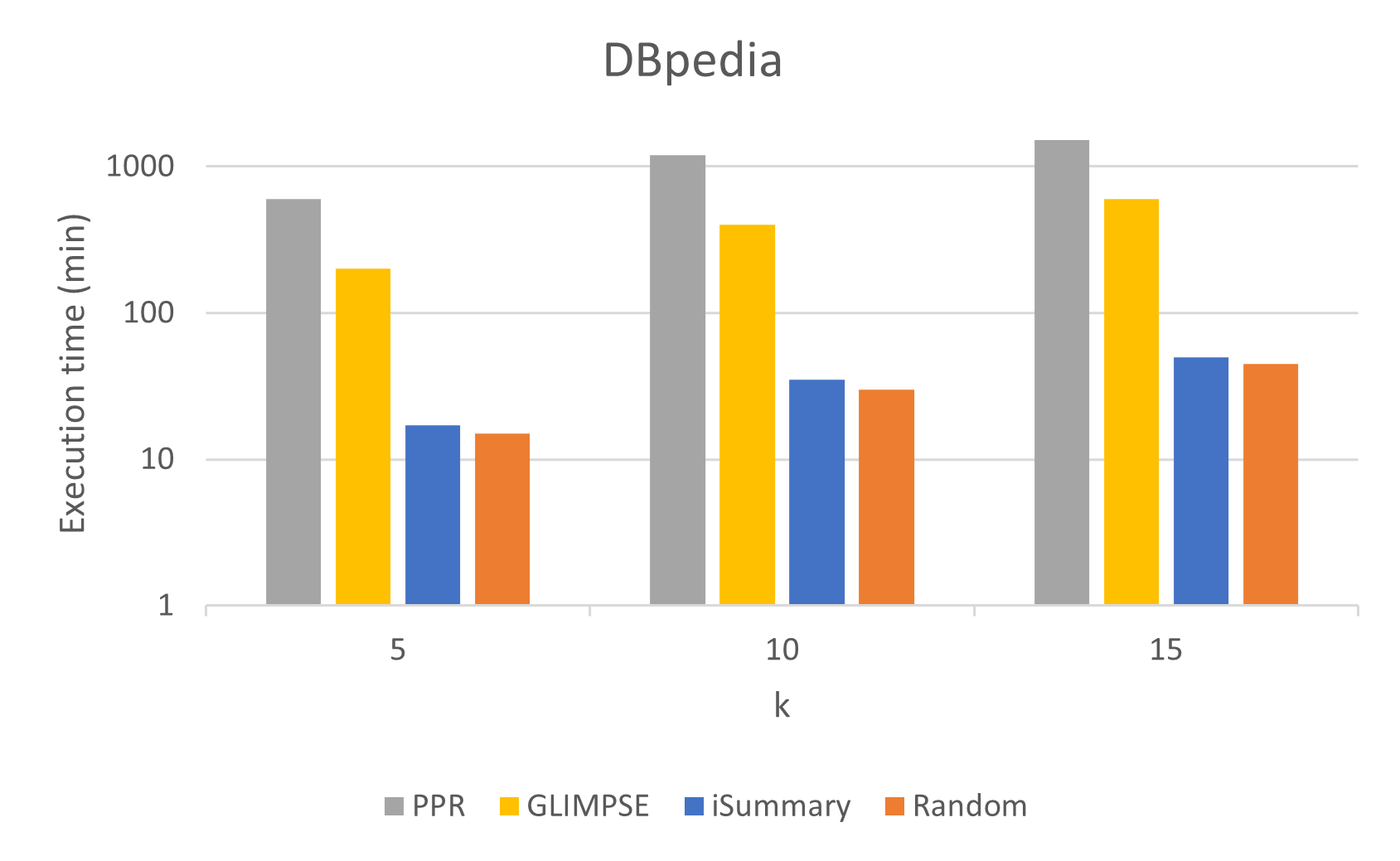}
    \caption{Execution time for the various $k$ and algorithms}
    \label{fig:execx}
\end{figure} 

\subsection{Comparing Execution Time} 
The average execution times for the various algorithms, for different $\kappa$ are presented in Figure~\ref{fig:execx}. We only present results for DBpedia as it is the only dataset that all competitors are able to run. As shown, approaches relying on the KG (PPR and GLIMPSE) to calculate the summary require one order of magnitude more execution time than the ones relying on query logs. iSummary is just a bit slower than Random showing that linking the $k$ nodes using queries has a minimal impact on query execution, but highly improves summaries' quality. Further, we can observe that as the $k$ grows, all algorithms require more time to identify and link more nodes. Overall, however, iSummary is only 0.13 times slower than Random, 14 times faster than GLIMPSE, and 40 times faster than PPR, however dominating all baselines in terms of coverage.

\section{Related Work} 

In this Section, we focus on personalized, structural, and non-quotient summaries and we present related works. For a complete overview of the works in the area, the interested reader is forwarded to relevant surveys available in the domain \cite{vcebiric2019summarizing}, \cite{DBLP:journals/vldb/Kellou-MenouerK22}.

Among the first works that focused on generating personalized non-quotients is \cite{zhang2007ontology}, that returns a summary of an RDF Sentence Graph. An RDF Sentence Graph is a weighted, directed graph where each vertex represents an RDF sentence. A link between two sentences exists if an object of one sentence belongs to another sentence as well. The creation of a sentence graph is customized by domain experts, who provide as input the desired summary size, and their navigation preferences, i.e. weights in the links they are most interested in. 
In Queiroz-Sousaet al. \cite{alzogbi2013similar}, on the other hand, the authors try to combine user preferences with the degree of centrality and the closeness to calculate the importance of a node, and then they use an algorithm to find paths that include the most important nodes in the final graph. However, in both these approaches, incorporating user preferences is neither explored in detail nor evaluated.

GLIMPSE \cite{DBLP:conf/icdm/SafaviBFMMK19} is the most relevant work to our approach and focuses on constructing personalized summaries of KG containing only the facts most relevant to individuals' interests. However, they require from the user to provide a set of relevant queries that would like relevant information to be included in the summary, whereas their algorithms are directly executed on the KG. However, it is difficult for a user to provide these queries and although the corresponding algorithm has linear complexity in the number of edges in the KG, still faces scalability problems. As shown is not able to run for big KGs.

Finally, there is latest approach named WBSUM~\cite{DBLP:conf/ssdbm/VassiliouTPK21} 
 which exploits query logs for constructing KG summaries. However, those are static, generic, and not personalized. 

Overall, our work is the first, structural, non-quotient, workload-based personalized summarization method. 
%However, we do not rely on generic centrality measures. 
Our work accepts minimal user input, and exploits query workloads to generate high-quality summaries. Further our algorithm is linear in the number of queries available and as such efficient and scalable.
%Further we do not focus only on the schema graph, but we treat both schema and instances the same, considering only their appearance in user queries. We argue that user queries can be more objective than centrality measures on identifying the most relevant parts for a given input by the user. 

\section{Conclusions} 
In this paper, we present a summarization method able to construct personalized, workload-based, semantic summaries with high quality. We formulate the problem of $\lambda/\kappa$-Personalized summaries and provide an elegant algorithm for resolving it, linear in the number of queries available in the query logs with theoretical guarantees. Our algorithm effectively identifies different weight assignments for different inputs and is able to efficiently and effectively identify how to link the selected nodes based on the available queries. 

We experimentally show that even 5k queries are enough for generating high-quality summaries (10\% of the query log in DBpedia) and we compare our approach with several baselines. We demonstrate that our approach strictly dominates all baselines in terms of query coverage (20-50\% better coverage than Random and 33-56\% better coverage than PPR and 40\% better coverage than GLIMPSE) and it is highly efficient being orders of magnitude faster than relevant approaches working directly on the KG.

\mpara{Future Work.} 
As future work, we intend to explore alternative methods for linking the $\kappa$ nodes to be used in the summary by exploiting the original data graph. The graph could be queried just once at the end for replacing the variables with the actual resources from the KG. This would introduce minimal overhead for querying the original graph and it would be quicker and possibly more effective than searching again all queries for filling the missing variables.

Another really interesting direction is to study how personalized summaries change over time for specific user input. As users' interests drift in time and we only require 4k-10k queries for generating high-quality summaries, it would be interesting to identify how the personalized summaries also change, considering that queries focus changes through time due to specific events, disasters, seasonality, or occasions.

Finally, as $\lambda/\kappa$-Personalized summaries are not unique, introducing the element of diversity would be interesting so that the users are not always presented with the same personalized summary.

\section*{Acknowledgments}
This research project was supported by the Hellenic Foundation
for Research and Innovation (H.F.R.I.) under the “2nd Call for H.F.R.I. Research Projects to support Post-Doctoral Researchers” (iQARuS Project No 1147).

\bibliographystyle{splncs04}
\bibliography{main}

\end{document}